\pdfoutput=1
\documentclass[10pt,fleqn]{article}

\usepackage{amsmath,amssymb,amsthm,amsfonts,mathtools,thm-restate}
\usepackage{fullpage}
\usepackage{authblk}
\usepackage[indent,parfill]{parskip}
\usepackage{indentfirst}

\usepackage[looser]{newpxtext}
\usepackage[varbb,varg,cmbraces,smallerops]{newpxmath}
\usepackage[style=alphabetic,natbib=true,sortcites=true,backref=true,maxnames=99,maxalphanames=99,maxcitenames=99]{biblatex}

\DefineBibliographyStrings{english}{%
  backrefpage = {$\Lsh$~p.},
  backrefpages = {$\Lsh$~pp.},
}
\usepackage[colorlinks=true,citecolor=blue!65!black,linkcolor=red!75!black]{hyperref}
\usepackage[nameinlink]{cleveref}

\usepackage[basic,bold,disableredefinitions]{complexity}

\usepackage{booktabs,enumitem}
\usepackage{ascmac}

\usepackage{comment}
\usepackage{url}

\usepackage{mleftright}
\mleftright

\usepackage[colorinlistoftodos,prependcaption]{todonotes}

\usepackage{xspace}
\usepackage{bm}

\usepackage{lineno}

\crefname{theorem}{Theorem}{Theorems}
\crefname{lemma}{Lemma}{Lemmas}
\crefname{claim}{Claim}{Claims}
\crefname{corollary}{Corollary}{Corollaries}
\crefname{remark}{Remark}{Remarks}
\crefname{observation}{Observation}{Observations}
\crefname{hypothesis}{Hypothesis}{Hypotheses}
\crefname{definition}{Definition}{Definitions}
\crefname{problem}{Problem}{Problems}

\crefname{appendix}{Appendix}{Appendices}
\crefname{section}{Section}{Sections}
\crefname{equation}{Eq.}{Eqs.}
\crefname{figure}{Figure}{Figures}
\crefname{table}{Table}{Tables}

\renewcommand{\geq}{\geqslant}
\renewcommand{\leq}{\leqslant}
\renewcommand{\phi}{\varphi}
\renewcommand{\epsilon}{\varepsilon}
\renewcommand{\bar}{\overline}

\newcommand{\prb}[1]{\textsc{#1}\xspace}

\renewcommand{\vec}[1]{\mathbf{\bm{#1}}}

\newcommand{\isep}{\mathrel{..}\nobreak}
\newcommand{\reco}{\leftrightsquigarrow}

\DeclareMathOperator{\val}{\mathsf{val}}
\DeclareMathOperator{\bigO}{\mathcal{O}}

\let\polylog\relax\DeclareMathOperator*{\polylog}{\mathrm{polylog}}

\newcommand{\ini}{\mathsf{ini}}
\newcommand{\tar}{\mathsf{tar}}
\newcommand{\ttt}{\mathsf{t}}
\newcommand{\TTT}{\mathsf{T}}
\newcommand{\asg}{\psi}
\newcommand{\sqasg}{\Psi}

\usepackage{algorithmicx}
\usepackage[ruled]{algorithm}
\usepackage[noend]{algpseudocode}
\algnewcommand\And{\; \textbf{and} \;}
\algnewcommand\Or{\; \textbf{or} \;}
\algnewcommand\To{\; \textbf{to} \;}
\algnewcommand\Continue{\textbf{continue}}
\algnewcommand\Not{\textbf{not}}

\algnewcommand{\algalign}[1]{\parbox[t]{\dimexpr\linewidth-\algorithmicindent}{#1\strut}}
\algrenewcommand\textproc{\textsl}

\newcommand{\bbN}{\mathbb{N}}

\numberwithin{equation}{section}

\newtheorem{theorem}{Theorem}[section]
\newtheorem{lemma}[theorem]{Lemma}

\newtheorem{remark}[theorem]{Remark}
\newtheorem{claim}[theorem]{Claim}

\newtheorem{observation}[theorem]{Observation}

\theoremstyle{definition}
\newtheorem{definition}[theorem]{Definition}

\title{On Approximate Reconfigurability of Label Cover}
\author{Naoto Ohsaka\thanks{CyberAgent, Inc., Tokyo, Japan. \href{mailto:ohsaka\_naoto@cyberagent.co.jp}{\texttt{ohsaka\_naoto@cyberagent.co.jp}}; \href{mailto:naoto.ohsaka@gmail.com}{\texttt{naoto.ohsaka@gmail.com}}
}}
\date{\today}

\begin{document}
\maketitle
\begin{abstract}Given a two-prover game $G$ and its two satisfying labelings $\asg_\ini$ and $\asg_\tar$,
the \prb{Label Cover Reconfiguration} problem asks whether
$\asg_\ini$ can be transformed into $\asg_\tar$
by repeatedly changing the label of a single vertex while preserving any intermediate labeling satisfying $G$.
We consider its optimization version by relaxing the feasibility of labelings,
referred to as \prb{Maxmin Label Cover Reconfiguration}:
We are allowed to pass through any \emph{non-satisfying} labelings,
but required to maximize the ``soundness error,'' which is defined as the  \emph{minimum} fraction of satisfied edges during transformation from $\asg_\ini$ to $\asg_\tar$.
Since the parallel repetition theorem of {Raz} (SIAM~J.~Comput.,~1998)~\cite{raz1998parallel},
which implies $\NP$-hardness of approximating \prb{Label Cover} within any constant factor, 
gives strong inapproximability results for many $\NP$-hard problems,
one may think of
using \prb{Maxmin Label Cover Reconfiguration} to derive inapproximability results for reconfiguration problems.
We prove the following results on \prb{Maxmin Label Cover Reconfiguration},
which display different trends from those of
\prb{Label Cover} and the parallel repetition theorem:
\begin{itemize}
    \item \prb{Maxmin Label Cover Reconfiguration}
        can be approximated within a factor of $\frac{1}{4} - o(1)$ for some restricted graph classes, including
        biregular graphs,
        balanced bipartite graphs with no isolated vertices, and
        superconstant average degree graphs.
    
    \item A ``naive'' parallel repetition of \prb{Maxmin Label Cover Reconfiguration}
        does not decrease the soundness error for \emph{every} two-prover game.
    
    \item \prb{Label Cover Reconfiguration} on \emph{projection games} can be decided in polynomial time.
\end{itemize}
Our results suggest that
a reconfiguration analogue of the parallel repetition theorem is unlikely.
\end{abstract}

\section{Introduction}
\subsection{Background}

In \emph{reconfiguration problems} \cite{ito2011complexity},
given a pair of feasible solutions for a combinatorial problem,
we seek for a step-by-step transformation from one to the other
while maintaining the feasibility of every intermediate solution.
Under the framework of reconfiguration due to \citet*{ito2011complexity},
numerous reconfiguration problems have been derived from classical combinatorial problems;
we refer the reader to the surveys by
\citet{nishimura2018introduction} and \citet{heuvel13complexity}.

In this article, we consider
reconfigurability of \prb{Label Cover} \cite{arora1997hardness} and its approximation.
Given a two-prover game $G$ and its two satisfying labelings $\asg_\ini$ and $\asg_\tar$,
the \prb{Label Cover Reconfiguration} problem asks whether
$\asg_\ini$ can be transformed into $\asg_\tar$
by repeatedly changing the label of a single vertex while preserving any intermediate labeling satisfying $G$.
Consider further the following \emph{optimization version} --- referred
to as \prb{Maxmin Label Cover Reconfiguration} --- by relaxing the feasibility of labelings:
We are allowed to pass through any \emph{non-satisfying} labelings,
but required to maximize the ``soundness error,''
which is defined as the \emph{minimum} fraction of satisfied edges during transformation from $\asg_\ini$ to $\asg_\tar$.
Solving it approximately,
we may find approximate reconfiguration for \prb{Label Cover Reconfiguration} consisting of \emph{almost-satisfying} labelings.

The recent work of the author~\cite{ohsaka2023gap} postulates
the \emph{Reconfiguration Inapproximability Hypothesis}\footnote{
This hypothesis has been resolved recently \cite{hirahara2024probabilistically,karthik2023inapproximability}. See also \cref{subsec:intro:followup}.
}
to give evidence that a bunch of reconfiguration problems such as \prb{Maxmin Label Cover Reconfiguration} are $\PSPACE$-hard to approximate within a factor of, say, $1-\epsilon$ for some $\epsilon \in (0,1)$.
One limitation of this approach is that
the value of such $\epsilon$ might be too small to
refute a $0.999\cdots 9$-approximation algorithm.
In the $\NP$ regime,
we can use the \emph{parallel repetition theorem} due to \citet{raz1998parallel}
to reduce the soundness error of a two-prover game efficiently and
give strong inapproximability results for many $\NP$-hard problems
\cite{feige1998threshold,hastad1999clique,hastad2001some,khanna2000hardness,bellare1998free,dinur2005hardness}.
The crux of this theorem is to imply
along with the PCP theorem~\cite{arora1998probabilistic,arora1998proof}
that 
for every $\epsilon \in (0,1)$, there is a finite alphabet $\Sigma$ such that
\prb{Label Cover} on $\Sigma$ is $\NP$-hard to approximate within a factor of $\epsilon$.\footnote{
The same result holds even if the underlying graph of a two-prover game is biregular.
}
Other parallel repetition theorems
are also known for 
quantum games \cite{yuen2016parallel},
$3$-player games \cite{girish2022polynomial}, and
probabilistically checkable proofs \cite{chiesa2024parallel}.
One might thus think of
a reconfiguration analogue of the parallel repetition theorem for \prb{Maxmin Label Cover Reconfiguration},
which would help in improving $\PSPACE$-hardness of approximation for reconfiguration problems.
Our contribution is to give evidence that such hopes are probably dashed.

\subsection{Our Results}
We present the following results on \prb{Maxmin Label Cover Reconfiguration},
which display different trends from those for \prb{Label Cover} and the parallel repetition theorem.
\begin{itemize}
    \item In \cref{sec:approx}, we prove that
        \prb{Maxmin Label Cover Reconfiguration}
        can be approximated within a factor of $\frac{1}{4}-o(1)$,
        \emph{independent of} the alphabet size $|\Sigma|$,
        for some restricted graph classes, including 
        biregular graphs,
        balanced bipartite graphs with no isolated vertices, and
        superconstant average degree graphs.
        Here, a bipartite graph $G = (X,Y,E)$
        is \emph{biregular}
        if every vertex on the same part has the same degree, and
        is \emph{balanced}
        if $X$ and $Y$ have the same size.
        We stress that \emph{biregular graphs} naturally arise by a standard reduction from 
        probabilistically checkable proofs to
        two-prover games (i.e., \prb{Label Cover} instances), see \citet{fortnow1994power} and \citet[Proposition~6.2]{radhakrishnan2007dinur};
        moreover, parallel repetition preserves biregularity.
        By the follow-up work \cite{karthik2023inapproximability},
        a $\frac{1}{4}$-factor was found to be not optimal;
        still, the present study is the first to show that
        \emph{\prb{Maxmin Label Cover Reconfiguration} can be approximated within a universal constant factor for every alphabet $\Sigma$ unlike \prb{Label Cover}}.
   
    \item In \cref{sec:pr}, we further show that
        a ``naive'' parallel repetition of \prb{Maxmin Label Cover Reconfiguration} does not decrease
        the soundness error for \emph{every} two-prover game.
        In particular,
        applying parallel repetition to (a gap version of) \prb{Maxmin Label Cover Reconfiguration} fails to amplify the gap between completeness and soundness.
    
    \item In \cref{sec:proj}, we develop
        a polynomial-time algorithm for deciding
        \prb{Label Cover Reconfiguration} on \emph{projection games}.
        The proof is based on a simple characterization of the reconfigurability between a pair of satisfying labelings.
        Note that
        if a \prb{Label Cover} instance is a projection game, then so is its parallel repetition,
        which has been a useful property in hardness of approximation \cite{lund1994hardness,hastad2001some}.
        Our polynomial-time solvability means that
        existing reductions relying on the projection property may not be reused or adapted
        in the reconfiguration regime.
\end{itemize}
Our results suggest that
a reconfiguration analogue of the parallel repetition theorem is unlikely; thus,
we should resort to a different approach
to derive an improved inapproximability factor for reconfiguration problems.

\subsection{Related Work}
The \prb{Label Cover} problem was first introduced by \citet{arora1997hardness}.
There are maximization and minimization versions of \prb{Label Cover}, and
both of them cannot be approximated within a factor of $2^{(\log n)^{1-\epsilon}}$
under $\NP \not\subseteq \DTIME(n^{\polylog(n)})$, where
$\epsilon \in (0,1)$ is any small constant and $n$ is the number of vertices.
\citet{dinur2004hardness} proved that
a minimization version of \prb{Label Cover} is $\NP$-hard to approximate
within a factor of $2^{(\log n)^{1-o(1)}}$,
improving upon \cite{arora1997hardness}.
\citet{kortsarz2001hardness}
formulated slightly different variants called \prb{Max Rep} and \prb{Min Rep},
which exhibit similar hardness and approximability results.
Our \prb{Maxmin Label Cover Reconfiguration} problem can also be thought of as
a reconfiguration analogue of \prb{Max Rep}.
Several approximation algorithms for \prb{Label Cover} have been developed
\cite{elkin2007hardness,peleg2007approximation,charikar2011improved,manurangsi2017improved}.
The current best approximation factor for the maximization version is $\bigO(n^{0.233})$ due to \citet{chlamtavc2017approximation} (to the best of our knowledge).

\subsection{Follow-up Work}
\label{subsec:intro:followup}
After an early draft of this paper was uploaded on arXiv, 
the following progress has been made.
The author demonstrated in \cite{ohsaka2024gap} that 
\prb{Maxmin Label Cover Reconfiguration} is $\NP$-hard to approximate within a factor of $\frac{3}{4}+\epsilon$ and
its slight generalization called \prb{Maxmin 2-CSP Reconfiguration}
is $\PSPACE$-hard to approximate within a factor of $0.9942$.
The proof of the latter adapts \citeauthor{dinur2007pcp}'s \emph{gap amplification} \cite{dinur2007pcp}.
Subsequently,
\citet{karthik2023inapproximability} proved matching lower and upper bounds
improving upon the present study significantly; i.e.,
$\NP$-hardness of $(\frac{1}{2}+\epsilon)$-factor approximation and
the existence of a $(\frac{1}{2}-\epsilon)$-factor approximation algorithm
for every $\epsilon > 0$.
Very recently,
\citet{hirahara2024probabilistically,karthik2023inapproximability}
independently gave the proof of the Reconfiguration Inapproximability Hypothesis \cite{ohsaka2023gap}
(thereby resolving an open problem of \citet{ito2011complexity}); thus,
\prb{Maxmin Label Cover Reconfiguration} is $\PSPACE$-hard to approximate \emph{unconditionally}.
\citet{hirahara2024optimal} took a different approach (e.g., FGLSS reduction \cite{feige1996interactive})
from parallel repetition to show that
a minimization version of \prb{Label Cover Reconfiguration} and
\prb{Minmax Set Cover Reconfiguration} are $\PSPACE$-hard to approximate within a factor of $2-o(1)$,
which matches an upper bound of \cite{ito2011complexity}.

\subsection{Preliminaries}
For two integers $m,n \in \bbN$ with $m \leq n$,
let $ [n] \coloneq \{1, 2, \ldots, n\} $ and
$[m\isep n] \coloneq \{m,m+1,\ldots, n-1,n\}$.
For a statement $P$, $\llbracket P \rrbracket$ is $1$ if $P$ is true, and $0$ otherwise.
We formally define \prb{Label Cover Reconfiguration} and its optimization version.
A \emph{constraint graph} is defined as a tuple $G=(V,E,\Sigma,\Pi)$, where
\begin{itemize}
    \item $(V,E)$ is an undirected graph called the \emph{underlying graph},
    \item $\Sigma$ is a finite set of labels called the \emph{alphabet}, and
    \item $\Pi = (\pi_e)_{e \in E}$ is a collection of binary \emph{constraints} and
    each constraint $\pi_e \subseteq \Sigma^e$ for edge $e \in E$
    consists of pairs of admissible labels that
    endpoints of $e$ can take.
\end{itemize}
If the underlying graph of $G$ is bipartite,
$G$ is called a \emph{two-prover game} or simply \emph{game}.
We write $G = (X,Y,E,\Sigma,\Pi)$ to stress that
the underlying graph of $G$ is a bipartite graph with bipartition $(X,Y)$.
Unless otherwise specified,
we use $G$ to represent a game.
Moreover, $G$ is said to be a \emph{projection game} if 
every constraint $\pi_{(x,y)}$ has a projection property; i.e.,
each label $\beta \in \Sigma$ for $y$ has a unique label $\alpha \in \Sigma$ for $x$ 
such that $(\alpha, \beta) \in \pi_{(x,y)}$.
Such $\alpha$ is denoted $\pi_{(x,y)}(\beta)$.

A \emph{labeling} for a constraint graph $G$
is a mapping $\asg \colon V \to \Sigma$ that assigns a label of $\Sigma$ to each vertex of $V$.
We say that $\asg$ \emph{satisfies} edge $e = (v,w) \in E$ (or constraint $\pi_e$) if
$(\asg(v), \asg(w)) \in \pi_e$, and
$\asg$ \emph{satisfies} $G$ if it satisfies all edges of $G$.
For two satisfying labelings $\asg_\ini$ and $\asg_\tar$ for $G$,
a \emph{reconfiguration sequence} from $\asg_\ini$ to $\asg_\tar$
is any sequence of labelings starting from $\asg_\ini$ and ending with $\asg_\tar$ such that
each labeling is obtained from the previous one by changing the label of a single vertex; i.e., 
any two neighboring labelings differ in exactly one vertex.
The \prb{Label Cover Reconfiguration} problem is defined as follows:

\begin{itembox}[l]{\prb{Label Cover Reconfiguration}}
\begin{tabular}{ll}
    \textbf{Input:}
    & a satisfiable game $G$ and
    its two satisfying labelings $\asg_\ini$ and $\asg_\tar$.
    \\
    \textbf{Question:}
    & is there a reconfiguration sequence from $\psi _{\mathsf{ini}}$ to $\psi _{\mathsf{tar}}$ consisting of satisfying labelings?
\end{tabular}
\end{itembox}

\prb{Label Cover Reconfiguration} is known to be $\PSPACE$-complete, see, e.g., \cite{gopalan2009connectivity,ito2011complexity}.
We then proceed to an optimization version \cite{ito2011complexity} of \prb{Label Cover Reconfiguration}, in which
we are allowed to touch \emph{non-satisfying} labelings.
For a game $G=(V,E,\Sigma,\Pi)$ and its labeling $\asg \colon V \to \Sigma$,
let $\val_G(\asg)$ denote the fraction of edges satisfied by $\asg$; namely,
\begin{align}
    \val_G(\asg) \coloneq \frac{1}{|E|}
        \left|\left\{e \in E \Bigm| \asg \text{\footnotesize{ satisfies }} e \right\}\right|.
\end{align}
For any reconfiguration sequence
$\sqasg = \langle \asg^{(0)}, \ldots, \asg^{(\TTT)} \rangle$,
let $\val_G(\sqasg)$ denote the minimum fraction of satisfied edges
over all $\asg^{(\ttt)}$'s in $\sqasg$; namely,
\begin{align}
    \val_G(\sqasg) \coloneq \min_{\asg^{(\ttt)} \in \sqasg} \val_G(\asg^{(\ttt)}).
\end{align}
Then, \prb{Maxmin Label Cover Reconfiguration} is defined as the following optimization problem:

\begin{itembox}[l]{\prb{Maxmin Label Cover Reconfiguration}}
\begin{tabular}{ll}
    \textbf{Input:}
    & a satisfiable game $G$ and
    its two satisfying labelings $\asg_\ini$ and $\asg_\tar$.
    \\
    \textbf{Question:}
    & maximize $\val_G(\sqasg)$ subject to that $\sqasg$ is a reconfiguration sequence from $\asg_\ini$ to $\asg_\tar$.
\end{tabular}
\end{itembox}
Note that the present definition does not request an actual reconfiguration sequence,
which may be of exponential length.

\section{Nearly $\frac{1}{4}$-factor Approximation}
\label{sec:approx}
In this section, we show that \prb{Maxmin Label Cover Reconfiguration} is
approximately reconfigurable within a factor of $\frac{1}{4}-o(1)$ for some graph classes.
For a game $G$ and its two labelings $\asg_\ini$ and $\asg_\tar$,
let $\val_G(\asg_\ini \reco \asg_\tar)$ denote the maximum value of $\val_G(\sqasg)$
over all possible reconfiguration sequences $\sqasg$ from $\asg_\ini$ to $\asg_\tar$; namely,
\begin{align}
    \val_G(\asg_\ini \reco \asg_\tar)
    \coloneq \max_{\sqasg = \langle \asg_\ini, \ldots, \asg_\tar \rangle} \val_G(\sqasg).
\end{align}

\begin{theorem}
\label{thm:approx:main}
For a game $G$ over $m$ edges and
its two satisfying labelings $\asg_\ini$ and $\asg_\tar$,
the following hold:
\begin{enumerate}[label=(\arabic*)]
    \item If $G$ is biregular with bipartition $(X,Y)$ (i.e., every vertex on the same part has the same degree), then
        $\val_G(\asg_\ini \reco \asg_\tar) \geq \frac{1}{4}\left(1-\frac{1}{\min\{|X|,|Y|\}}\right)$.
    \item If $G$ is balanced (i.e., its two parts have the same size) and has no isolated vertices, then
    $\val_G(\asg_\ini \reco \asg_\tar) \geq \frac{1}{4}\left(1-\frac{1}{\sqrt{m}}\right)$.
    \item If the average degree of $G$ is at least $\frac{6}{\epsilon}$ for $\epsilon \in (0,\frac{1}{4})$, then
    $\val_G(\asg_\ini \reco \asg_\tar) \geq \frac{1}{4}-\epsilon$.
    Moreover, the same result holds even if
    the underlying graph of $G$ is a general (i.e., non-bipartite) graph.
\end{enumerate}
Moreover, an explicit reconfiguration sequence for the former two cases can be found in polynomial time.
\end{theorem}
The proof of \cref{thm:approx:main} relies on the following claim.

\begin{claim}
\label{lem:approx:via}
For a constraint graph $G = (V,E,\Sigma,\Pi)$ and its two satisfying labelings $\asg_\ini$ and $\asg_\tar$,
it holds that
\begin{align}
    \val_G(\asg_\ini \reco \asg_\tar) \geq 
    \max_{\emptyset \subsetneq S \subsetneq V}
    \min\left\{ \frac{|E[S]|}{|E|}, \frac{|E[V \setminus S]|}{|E|} \right\},
\end{align}
where $E[S]$ is the edge set of a subgraph of $G$ induced by vertex set $S$.
\end{claim}
\begin{proof}
For a vertex set $S$, $\emptyset \subsetneq S \subsetneq V$,
we define a labeling
$\breve{\asg} \colon V \to \Sigma$ as follows:
\begin{align}
    \breve{\asg}(v) \coloneq
    \begin{cases}
        \asg_\ini(v) & \text{if } v \in S, \\
        \asg_\tar(v) & \text{if } v \notin S.
    \end{cases}
\end{align}
Consider transforming $\asg_\ini$ into $\breve{\asg}$
by changing the label of vertex $v \notin S$ from $\asg_\ini(v)$ to $\asg_\tar(v)$ one by one.
Since the label of any vertex in $S$ has never been changed,
any edge of $E[S]$ is always satisfied during this transformation; i.e.,
$\val_G(\asg_\ini \reco \breve{\asg}) \geq \frac{|E[S]|}{|E|}$.
Similarly, it follows that
$\val_G(\breve{\asg} \reco \asg_\tar) \geq \frac{|E[V \setminus S]|}{|E|}$.
Consequently, we derive
\begin{align}
    \val_G(\asg_\ini \reco \asg_\tar)
    \geq \min\Bigl\{
        \val_G(\asg_\ini \reco \breve{\asg}), \val_G(\breve{\asg} \reco \asg_\tar)
    \Bigr\}
    \geq \min\left\{ \frac{|E[S]|}{|E|}, \frac{|E[V \setminus S]|}{|E|} \right\},
\end{align}
as desired.
\end{proof}
By \cref{lem:approx:via}, it suffices to find or show the existence of
a partition of the vertex set,
each of which contains $\gtrapprox \frac{|E|}{4}$ edges for each graph class.

\begin{remark}
    \cref{lem:approx:via} inherently cannot derive a better-than-$\frac{1}{4}$ approximation
    since for an $n$-vertex complete graph $G=(V,E)$,
    \begin{align}
        \max_{\emptyset \subsetneq S \subsetneq V} \min\left\{
            \frac{|E[S]|}{|E|}, \frac{|E[V \setminus S]|}{|E|}
        \right\}
        \leq \frac{{\frac{n}{2} \choose 2}}{{n \choose 2}}
        = \frac{1}{4} - o(1).
    \end{align}
\end{remark}

\subsection{Biregular Graphs}
Here, we partition the vertex set of a biregular graph by the following lemma,
which brings us the first statement of \cref{thm:approx:main}.

\begin{lemma}
\label{lem:approx:biregular}
For a biregular graph $G = (X,Y,E)$ over $m$ edges,
there exists a partition $(S,\bar{S})$ of $X \cup Y$ such that
\begin{align}
    \min\Bigl\{|E[S]|, |E[\bar{S}]|\Bigr\}
    \geq \frac{m}{4}\left(1-\frac{1}{\min\{|X|, |Y|\}} \right).
\end{align}
Moreover, such a partition can be found in polynomial time.
\end{lemma}
\begin{proof}
Let $d_X \coloneq \frac{m}{|X|}$ and $d_Y \coloneq \frac{m}{|Y|}$
denote the degree of vertices of $X $ and $Y$, respectively.
Hereafter, for any pair of subsets $S \subseteq X$ and $T \subseteq Y$,
we define $e(S,T)$ as the number of edges of $G$ between $S$ and $T$.
Let
$(S,\bar{S})$ be any partition of $X$ such that $\min\{|S|, |\bar{S}|\} \geq \left\lfloor \frac{|X|}{2} \right\rfloor$ and
$(T,\bar{T})$ any partition of $Y$ such that $\min\{|T|, |\bar{T}|\} \geq \left\lfloor \frac{|Y|}{2} \right\rfloor$.
Define 
\begin{align}
    \theta \coloneq \frac{m}{4} \left( 1-\frac{1}{\min\{|X|, |Y|\}} \right).
\end{align}
By case analysis, we show that either of
$(S \cup T, \bar{S}\cup \bar{T})$ or $(S \cup \bar{T}, \bar{S} \cup T)$
is a desired partition.
\begin{description}
\item[\textbf{(Case 1)}]
    $e(S,T) \leq \theta$: 
    $e(S,\bar{T})$ and $e(\bar{S},T)$ can be bounded as follows:
    \begin{align}
    \begin{aligned}
        e(S,\bar{T})
        & = e(S,T\cup\bar{T}) - e(S,T)
        \geq \left\lfloor \frac{|X|}{2} \right\rfloor d_X - \theta \\
        & \geq \left(\frac{|X|-1}{2}\right) d_X - \theta
        \geq \frac{m}{2} \left(1-\frac{1}{|X|}\right) - \frac{m}{4} \left( 1-\frac{1}{\min\{|X|, |Y|\}} \right) \\
        & = \frac{m}{4} - m \left(\frac{1}{2|X|} - \frac{1}{4\min\{|X|,|Y|\}}\right)
        \geq \frac{m}{4} \left(1 - \frac{1}{\min\{|X|, |Y|\}}\right).
    \end{aligned}
    \end{align}
    \begin{align}
    \begin{aligned}
        e(\bar{S},T)
        & = e(S\cup\bar{S},T) - e(S,T)
        \geq \left\lfloor \frac{|Y|}{2} \right\rfloor d_Y - \theta \\
        & \geq \left(\frac{|Y|-1}{2}\right) d_Y - \theta
        \geq \frac{m}{2} \left(1-\frac{1}{|Y|}\right) - \frac{m}{4} \left( 1-\frac{1}{\min\{|X|, |Y|\}} \right) \\
        & = \frac{m}{4} - m \left(\frac{1}{2|Y|} - \frac{1}{4\min\{|X|,|Y|\}}\right)
        \geq \frac{m}{4} \left(1 - \frac{1}{\min\{|X|, |Y|\}}\right).
    \end{aligned}
    \end{align}
\item[\textbf{(Case 2)}]
    $e(\bar{S}, \bar{T}) \leq \theta$:
    Similarly to (Case 1), we obtain
    \begin{align}
        e(S, \bar{T}) \geq \frac{m}{4} \left(1 - \frac{1}{\min\{|X|, |Y|\}}\right)
        \text{ and }
        e(\bar{S}, T) \geq \frac{m}{4} \left(1 - \frac{1}{\min\{|X|, |Y|\}}\right).
    \end{align}
\item[\textbf{(Case 3)}]
    $e(S,T) > \theta$ and $e(\bar{S}, \bar{T}) > \theta$: We are already done.
\end{description}
Consequently, either
$\min\{e(S,T), e(\bar{S}, \bar{T})\}$
or
$\min\{e(S,\bar{T}), e(\bar{S}, T)\}$
is no less than $\theta$, implying that
\begin{align}
    \max\Biggr\{
        \min\Bigl\{ |E[S\cup T]|, |E[\bar{S}\cup\bar{T}]| \Bigr\},
        \min\Bigl\{ |E[S\cup \bar{T}]|, |E[\bar{S}\cup T]| \Bigr\}
    \Biggr\}
    \geq \theta,
\end{align}
completing the proof.
\end{proof}

\subsection{Balanced Bipartite Graphs}
Subsequently, for a balanced bipartite graph $G = (X,Y,E)$ with no isolated vertices,
we partition each $X$ and $Y$ in a particular well-balanced manner as follows:

\begin{lemma}
\label{lem:approx:balanced}
For a balanced bipartite graph $G = (X,Y,E)$ over $m$ edges with no isolated vertices such that $|X|=|Y|$,
there exists a partition $(S,\bar{S})$ of $X \cup Y$ such that
\begin{align}
    \min\Bigl\{|E[S]|, |E[\bar{S}]| \Bigr\} \geq \frac{m-\sqrt{m}}{4}.
\end{align}
Moreover, such a partition can be found in polynomial time.
\end{lemma}
\cref{lem:approx:balanced} leads to the second statement of \cref{thm:approx:main}.
Before proving \cref{lem:approx:balanced}, we show the following auxiliary lemma.

\begin{lemma}
\label{lem:approx:partition}
For any $n$ positive integers $x_1, \ldots, x_n$ in the range of $[n]$,
there exists a partition $(S,T)$ of $[n]$ such that
\begin{align}
    \min\left\{ \sum_{i \in S} x_i, \sum_{i \in T} x_i \right\} \geq \frac{m-\sqrt{m}}{2},
    \text{  where  }
    m \coloneq \sum_{i \in [n]} x_i.
\end{align}
Moreover, such $S$ and $T$ can be found in polynomial time.
\end{lemma}
\begin{proof}
The case of $n=1$ is trivial because $m=1$ and so $\frac{m-\sqrt{m}}{2} = 0$.
Hereafter, we assume that $n \geq 2$.
Denote $\vec{x}(S) \coloneq \sum_{i \in S} x_i$ for any $S \subseteq [n]$.
Consider the following naive greedy algorithm for \prb{Bin Packing}:

\begin{itembox}[l]{\textbf{Greedy algorithm}}
\begin{algorithmic}[1]
    \State sort $x_i$'s in descending order.
    \State define $S^{(0)} \coloneq \emptyset$ and $T^{(0)} \coloneq \emptyset$.
    \For {$i = 1 \To n$} 
        \If {$\vec{x}(S^{(i-1)}) < \vec{x}(T^{(i-1)})$}
            \State define $S^{(i)} \coloneq S^{(i-1)} \cup \{i\}$ and $T^{(i)} \coloneq T^{(i-1)}$.
        \Else
            \State define $S^{(i)} \coloneq S^{(i-1)}$ and $T^{(i)} \coloneq T^{(i-1)} \cup \{i\}$.
        \EndIf
    \EndFor
    \State \textbf{return} $S^{(n)}$ and $T^{(n)}$.
\end{algorithmic}
\end{itembox}

We will show
$
    \left| \vec{x}(S^{(n)}) - \vec{x}(T^{(n)}) \right| \leq \sqrt{m}.
$
Define
\begin{align}
    d \coloneq \frac{m}{n} \text{ and }
    \epsilon \coloneq \sqrt{\frac{d}{n}} = \sqrt{\frac{m}{n^2}}.
\end{align}
Suppose $x_i$'s are sorted in descending order.
Then, for every $k \in [n]$,
$x_1, \ldots, x_k$ are at least $x_k$ and
$x_{k+1}, \ldots, x_n$ are at least $1$.
Therefore, we have
\begin{align}
\label{eq:approx:partition:xk}
\begin{aligned}
& \underbrace{x_k \cdot k}_{\text{contribution of } x_1, \ldots, x_k}
+ \underbrace{1 \cdot (n-k)}_{\text{contribution of } x_{k+1}, \ldots, x_n} \leq m = nd \\
& \implies k(x_k - 1) \leq (d-1)n \\
& \implies x_k \leq \frac{d-1}{k}n + 1.
\end{aligned}
\end{align}
Observing easily that
$| \vec{x}(S^{(\lceil\epsilon n\rceil-1)}) - \vec{x}(T^{(\lceil\epsilon n\rceil-1)}) | \leq n$
due to the nature of the greedy algorithm,
we consider the following two cases:

\begin{description}
    \item[\textbf{(Case 1)}]
    If the absolute difference between
    $\vec{x}(S^{(\lceil \epsilon n \rceil-1)})$ and
    $\vec{x}(T^{(\lceil \epsilon n \rceil-1)})$
    is larger than $\vec{x}([\lceil \epsilon n \rceil \isep n])$,
    we would have added all of $[\lceil \epsilon n \rceil \isep n]$
    into either $S^{(n)}$ or $T^{(n)}$.
    Thus, the absolute difference between
    $\vec{x}(S^{(n)})$ and $\vec{x}(T^{(n)})$
    will be simply
    \begin{align}
        \left|\vec{x}(S^{(n)}) - \vec{x}(T^{(n)})\right|
        \leq
          \underbrace{
            \left| \vec{x}(S^{(\lceil \epsilon n \rceil-1)}) - \vec{x}(T^{(\lceil \epsilon n \rceil-1)})\right|
            }_{\leq n}
          - \underbrace{
          \vec{x}([\lceil \epsilon n \rceil \isep n])
          }_{\geq n-\lceil \epsilon n \rceil + 1}
          \leq \lceil \epsilon n \rceil - 1
          \leq \epsilon n.
    \end{align}
    
    \item[\textbf{(Case 2)}]
    Otherwise
    (i.e.,
        $| \vec{x}(S^{(\lceil\epsilon n\rceil-1)}) - \vec{x}(T^{(\lceil\epsilon n\rceil-1)}) | \leq \vec{x}([\lceil \epsilon n \rceil \isep n]$),
    the absolute difference between $\vec{x}(S^{(n)})$ and $\vec{x}(T^{(n)})$
    must be at most $x_{\lceil \epsilon n \rceil}$.
    Using \cref{eq:approx:partition:xk}, we can derive
    \begin{align}
        \left|\vec{x}(S^{(n)}) - \vec{x}(T^{(n)})\right|
        \leq
        x_{\lceil \epsilon n \rceil}
        \leq \frac{d-1}{\lceil \epsilon n \rceil}n + 1
        \leq \frac{d-1}{\epsilon}+1
        \leq \frac{d}{\epsilon},
    \end{align}
    where the last inequality holds because $\epsilon \in (0,1]$.
\end{description}

In either case, we obtain
\begin{align}
    \left|\vec{x}(S^{(n)}) - \vec{x}(T^{(n)})\right|
    \leq \max\left\{ \epsilon n, \frac{d}{\epsilon} \right\}
    = \sqrt{dn}
    = \sqrt{m}.
\end{align}
Consequently, we derive
\begin{align}
\begin{aligned}  
& 2 \min\Bigl\{\vec{x}(S^{(n)}), \vec{x}(T^{(n)})\Bigr\} + \sqrt{m}
\geq \min\Bigl\{\vec{x}(S^{(n)}), \vec{x}(T^{(n)})\Bigr\} + \max\Bigl\{\vec{x}(S^{(n)}), \vec{x}(T^{(n)})\Bigr\}
= m \\
\implies & \min\Bigl\{\vec{x}(S^{(n)}), \vec{x}(T^{(n)})\Bigr\} \geq \frac{m-\sqrt{m}}{2},
\end{aligned}
\end{align}
completing the proof.
\end{proof}

\begin{proof}[Proof of \cref{lem:approx:balanced}]
Let $d(v)$ denote the degree of vertex $v$ of $G$.
Hereafter, for any pair of subsets $S \subseteq X$ and $T \subseteq Y$,
we define $e(S,T)$ as the number of edges of $G$ between $S$ and $T$.
Since it holds that $1 \leq d(v) \leq |X|=|Y|$ for every vertex $v$,
we use \cref{lem:approx:partition} to construct a pair of
partitions $(S,\bar{S})$ of $X$ and $(T,\bar{T})$ of $Y$ such that
\begin{align}
    \min\left\{ \sum_{x \in S}d(x), \sum_{x \in \bar{S}}d(x) \right\}
    \geq \frac{m-\sqrt{m}}{2}
    \text{  and  }
    \min\left\{ \sum_{y \in T}d(x), \sum_{y \in \bar{T}}d(x) \right\}
    \geq \frac{m-\sqrt{m}}{2}.
\end{align}
Define $\theta \coloneq \frac{m-\sqrt{m}}{4}$.
By case analysis, we show that either of $(S\cup T, \bar{S}\cup \bar{T})$ or $(S\cup \bar{T}, \bar{S}\cup T)$
gives us a desired partition.
\begin{description}
    \item[\textbf{(Case 1)}] $e(S,T) \leq \theta$:
    $e(S,\bar{T})$ and $e(\bar{S},T)$ can be bounded from below as follows:
    \begin{align}
        e(S,\bar{T})
        & = e(S,T\cup\bar{T}) - e(S,T)
        \geq \sum_{x \in S} d(x) - \theta
        \geq \frac{m-\sqrt{m}}{4}. \\
        e(\bar{S},T)
        & =e(S\cup\bar{S}, T) - e(S,T)
        \geq \sum_{y \in T} d(y) - \theta
        \geq \frac{m-\sqrt{m}}{4}.
    \end{align}
    \item[\textbf{(Case 2)}] $e(\bar{S},\bar{T}) \leq \theta$:
    Similarly to (Case 1), we can show that
    \begin{align}
        e(S,\bar{T}) \geq \frac{m-\sqrt{m}}{4} \text{ and }
        e(\bar{S},T) \geq \frac{m-\sqrt{m}}{4}.
    \end{align}
    \item[\textbf{(Case 3)}] $e(S,T) > \theta$ and $e(\bar{S},\bar{T}) > \theta$:
    We are already done.
\end{description}
Consequently, either
$\min\{e(S,T), e(\bar{S}, \bar{T})\}$
or
$\min\{e(S,\bar{T}), e(\bar{S}, T)\}$
is no less than $\frac{m-\sqrt{m}}{4}$,
implying that
\begin{align}
    \max\Biggl\{
        \min\Bigl\{ |E[S\cup T]|, |E[\bar{S}\cup\bar{T}]| \Bigr\},
        \min\Bigl\{ |E[S\cup \bar{T}]|, |E[\bar{S}\cup T]| \Bigr\}
    \Biggr\}
    \geq \frac{m-\sqrt{m}}{4},
\end{align}
completing the proof.
\end{proof}

\subsection{High Average Degree Graphs}

We refer to \citet{kuhn2003partitions},
which along with \cref{lem:approx:via}
directly implies the third statement of \cref{thm:approx:main}.

\begin{theorem}[\protect{\citet[Corollary 18]{kuhn2003partitions}}]
For every graph $G=(V,E)$ of $m$ edges and average degree $\frac{6}{\epsilon}$ for $\epsilon \in (0,\frac{1}{4})$,
its vertex set $V$ can be partitioned into
$S$ and $T$ such that both $E[S]$ and $E[T]$ contain at least $m(\frac{1}{4}-\epsilon)$ edges.
\end{theorem}

\section{Parallel Repetition Does Not Decrease the Soundness Error}
\label{sec:pr}

We show that a ``naive'' parallel repetition of
\prb{Maxmin Label Cover Reconfiguration} does not decrease its optimal value.
The \emph{product} of two games is first introduced.
\begin{definition}
Let
$G_1 = (X_1,Y_1,E_1,\Sigma_1,\Pi_1 = (\pi_{1,e})_{e \in E_1})$ and 
$G_2 = (X_2,Y_2,E_2,\Sigma_2,\Pi_2 = (\pi_{2,e})_{e \in E_2})$ be
two games.
Then, the \emph{product} of $G_1$ and $G_2$, denoted $G_1 \otimes G_2$,
is defined as a new game
$(X_1 \times X_2, Y_1 \times Y_2, E, \Sigma_1 \times \Sigma_2, \Pi = (\pi_e)_{e \in E})$, where
\begin{align}
    E \coloneq \left\{
        ((x_1,x_2), (y_1,y_2)) \Bigm| (x_1,y_1) \in E_1 \text{ and } (x_2,y_2) \in E_2
    \right\},
\end{align}
and the constraint $\pi_e \subseteq (\Sigma_1 \times \Sigma_2)^e$
for each edge $e = ((x_1,x_2),(y_1,y_2)) \in E$ is defined as
\begin{align}
    \pi_e \coloneq \left\{
        ((\alpha_1,\alpha_2), (\beta_1,\beta_2)) \Bigm|
        (\alpha_1, \beta_1) \in \pi_{1,(x_1,y_1)} \text{ and }
        (\alpha_2, \beta_2) \in \pi_{2,(x_2,y_2)}
    \right\}.
\end{align}
\end{definition}

The \emph{$\rho$-fold parallel repetition} of a game $G$,
denoted $G^{\otimes \rho}$,
for any positive integer $\rho \in \bbN$ is defined as
\begin{align}
    G^{\otimes \rho} \coloneq \underbrace{G \otimes \cdots \otimes G}_{\rho \text{ times}}.
\end{align}
For a game $G$,
let $\val(G)$ denote the maximum fraction of edges satisfied by any possible labeling of $G$.
The parallel repetition theorem~\cite{raz1998parallel} states that 
for every game $G$ with $\val(G) \leq 1-\epsilon$ for some $\epsilon \in (0,1)$,
it holds that
$\val(G^{\otimes \rho}) \leq (1-\bar{\epsilon})^{\frac{\rho}{\log |\Sigma|}}$, 
where $\bar{\epsilon} \in (0,1)$ depends only on $\epsilon$.

We can think of a reconfiguration analogue for parallel repetition.
For two labelings
$\asg_1 \colon (X_1 \cup Y_1) \to \Sigma_1$ for $G_1=(X_1,Y_1,E_1,\Sigma_1,\Pi_1)$ and
$\asg_2 \colon (X_2 \cup Y_2) \to \Sigma_2$ for $G_2=(X_2,Y_2,E_2,\Sigma_2,\Pi_2)$,
the \emph{product labeling} of $\asg_1$ and $\asg_2$,
denoted $\asg_1 \otimes \asg_2$,
is defined as a labeling
$\asg \colon (X_1 \times X_2) \cup (Y_1 \times Y_2) \to \Sigma_1 \times \Sigma_2$ such that
\begin{align}
    \asg(v_1, v_2) \coloneq (\asg_1(v_1), \asg_2(v_2))
        \text{ for all } (v_1, v_2) \in (X_1 \times X_2) \cup (Y_1 \times Y_2).
\end{align}
We write $\asg^{\otimes \rho}$ for denoting
$
    \underbrace{\asg \otimes \cdots \otimes \asg}_{\rho \text{ times}}.
$
It is easy to see that
$\val_G(\asg_\ini \reco \asg_\tar) = 1$ implies
$\val_{G^{\otimes \rho}}(\asg_\ini^{\otimes \rho} \reco \asg_\tar^{\otimes \rho}) = 1$.
So, given that $\val_G(\asg_\ini \reco \asg_\tar) < 1$,
does the value of the $\rho$-fold parallel repetition of a game
decreases as the increase of $\rho$?
Unfortunately, the answer is negative.

\begin{observation}
Let $G$ be a satisfiable game and
$\asg_\ini$ and $\asg_\tar$ be its two satisfying labelings.
Then, for every positive integer $\rho \in \bbN$,
it holds that
\begin{align}
\label{eq:pr}
    \val_{G^{\otimes\rho}}\left(\asg_\ini^{\otimes\rho} \reco \asg_\tar^{\otimes\rho}\right)
    \geq \val_G(\asg_\ini \reco \asg_\tar).
\end{align}
\end{observation}
\begin{proof}
Let $G = (X,Y,E,\Sigma,\Pi)$ be a satisfiable game.
Given a reconfiguration sequence
$\sqasg = \langle \asg^{(0)} \ldots, \asg^{(\TTT)} \rangle$
from $\asg_\ini$ to $\asg_\tar$ such that
$\val_G(\sqasg) = \val_G(\asg_\ini \reco \asg_\tar)$,
we will show that for every $k \in [\rho]$,
\begin{align}
    \val_{G^{\otimes \rho}}(
        \asg_\ini^{\otimes k} \otimes \asg_\tar^{\otimes \rho-k} \reco
        \asg_\ini^{\otimes k-1} \otimes \asg_\tar^{\otimes \rho-k+1}
    )
    \geq
    \val_G(\asg_\ini \reco \asg_\tar),
\end{align}
which implies \cref{eq:pr}.
To this end, we shall bound the following value for each $\ttt \in [\TTT]$:
\begin{align}
\label{eq:pr:bound}
    \val_{G^{\otimes \rho}}(
        \asg_\ini^{\otimes k-1} \otimes \asg^{(\ttt-1)} \otimes \asg_\tar^{\otimes \rho-k}
        \reco
        \asg_\ini^{\otimes k-1} \otimes \asg^{(\ttt)} \otimes \asg_\tar^{\otimes \rho-k}
    ).
\end{align}
Hereafter, we
fix $k \in [\rho]$ and $\ttt \in [\TTT]$.
Suppose $\asg^{(\ttt-1)}$ and $\asg^{(\ttt)}$ differ in some $x^\star \in X$.
\cref{eq:pr:bound} can be shown similarly
when $\asg^{(\ttt-1)}$ and $\asg^{(\ttt)}$ differ in some $y^\star \in Y$.
Construct
a reconfiguration sequence $\sqasg^{(k,\ttt)}$ from
$\asg_\ini^{\otimes k-1} \otimes \asg^{(\ttt-1)} \otimes \asg_\tar^{\otimes \rho-k}$
to
$\asg_\ini^{\otimes k-1} \otimes \asg^{(\ttt)} \otimes \asg_\tar^{\otimes \rho-k}$
by changing the label of each vertex of $G^{\otimes \rho}$
in which the two labelings differ.
Let $\psi^\bullet \colon X^\rho \cup Y^\rho \to \Sigma^\rho$ be any labeling of $\sqasg^{(k,\ttt)}$,
whose value is equal to

\begin{align}
\begin{aligned}    
    \val_{G^{\otimes \rho}}(\asg^\bullet)
    & = \frac{1}{|E(G^{\otimes \rho})|}
    \sum_{\substack{
        x_1, \ldots, x_\rho \in X, \\
        y_1, \ldots, y_\rho \in Y: \\
        (x_i, y_i) \in E \forall i
    }}
    \left\llbracket
        \bigwedge_{1 \leq i \leq \rho} \Bigl(\asg^\bullet(x_1, \ldots, x_\rho)_i, \asg^\bullet(y_1, \ldots, y_\rho)_i\Bigr) \in \pi_{(x_i,y_i)}
    \right\rrbracket \\
    & = \frac{1}{|E(G^{\otimes \rho})|}
    \sum_{\substack{
        x_1, \ldots, x_\rho \in X
    }}
    \underbrace{
    \sum_{\substack{
        y_1, \ldots, y_\rho \in Y: \\
        (x_i, y_i) \in E \forall i
    }}
    \left\llbracket
        \bigwedge_{1 \leq i \leq \rho} \Bigl(\asg^\bullet(x_1, \ldots, x_\rho)_i, \asg^\bullet(y_1, \ldots, y_\rho)_i\Bigr) \in \pi_{(x_i,y_i)}
    \right\rrbracket
    }_{\heartsuit \coloneq}.
\end{aligned}
\end{align}
Observe that
$\asg_\ini^{\otimes k-1} \otimes \asg^{(\ttt-1)} \otimes \asg_\tar^{\otimes \rho-k}$
and
$\asg_\ini^{\otimes k-1} \otimes \asg^{(\ttt)} \otimes \asg_\tar^{\otimes \rho-k}$
differ on vertex $(x_1, \ldots, x_\rho)$ of $X^\rho$
if and only if $x_k = x^\star$, and
they agree on every vertex $(y_1, \ldots, y_\rho)$ of $Y^\rho$; namely,
\begin{align}
    \asg(x_1, \ldots, x_\rho) & \text{ is }
    \left\{
    \begin{array}{c}
        \left(\asg_\ini(x_1), \ldots, \asg_\ini(x_{k-1}), \asg^{(\ttt-1)}(x_k), \asg_\tar(x_{k+1}), \ldots, \asg_\tar(x_\rho)\right)
        \\ \text{or} \\
        \left(\asg_\ini(x_1), \ldots, \asg_\ini(x_{k-1}), \asg^{(\ttt)}(x_k), \asg_\tar(x_{k+1}), \ldots, \asg_\tar(x_\rho)\right)
    \end{array}
    \right.
\end{align}
for all $(x_1, \ldots, x_\rho) \in X^\rho$ and
\begin{align}
    \asg(y_1, \ldots, y_\rho)
    \text{ is } \Bigl(\asg_\ini(y_1), \ldots, \asg_\ini(y_{k-1}), \underbrace{\asg^{(\ttt)}(y_k)}_{= \asg^{(\ttt-1)}(y_k)}, \asg_\tar(y_{k+1}), \ldots, \asg_\tar(y_\rho)\Bigr)
\end{align}
for all $(y_1, \ldots, y_\rho) \in Y^\rho$.
Since $\asg_\ini$ and $\asg_\tar$ satisfy all edges of $G$ by assumption,
we can evaluate the value of $\heartsuit$ as follows:

\begin{align}
\begin{aligned}
    \heartsuit & \geq \min\left\{
        \sum_{\substack{
            y_1, \ldots, y_\rho \in Y: \\
            (x_i, y_i) \in E \forall i
        }}
        \Biggl\llbracket
            \bigwedge_{1 \leq i \leq k-1} \underbrace{\asg_\ini \text{\footnotesize{ satisfies }} (x_i, y_i)}_{\text{always true}}
        \Biggr\rrbracket \cdot
        \Biggl\llbracket
            \asg^{(\ttt-1)} \text{\footnotesize{ satisfies }} (x_k, y_k)
        \Biggr\rrbracket \cdot
        \Biggl\llbracket
            \bigwedge_{k+1 \leq i \leq \rho} \underbrace{\asg_\tar \text{\footnotesize{ satisfies }} (x_i, y_i)}_{\text{always true}}
        \Biggr\rrbracket \right. ,
     \\
     &\qquad\qquad \left.
        \sum_{\substack{
            y_1, \ldots, y_\rho \in Y: \\
            (x_i, y_i) \in E \forall i
        }}
        \Biggl\llbracket
            \bigwedge_{1 \leq i \leq k-1} \underbrace{\asg_\ini \text{\footnotesize{ satisfies }} (x_i, y_i)}_{\text{always true}}
        \Biggr\rrbracket \cdot
        \Biggl\llbracket
            \asg^{(\ttt)} \text{\footnotesize{ satisfies }} (x_k, y_k)
        \Biggr\rrbracket \cdot
        \Biggl\llbracket
            \bigwedge_{k+1 \leq i \leq \rho} \underbrace{\asg_\tar \text{\footnotesize{ satisfies }} (x_i, y_i)}_{\text{always true}}
        \Biggr\rrbracket
    \right\} \\
    & = \min\left\{
        \sum_{\substack{
            y_1, \ldots, y_\rho: \\
            (x_i, y_i) \in E \forall i
        }} \left\llbracket
            \asg^{(\ttt-1)} \text{\footnotesize{ satisfies }} (x_k, y_k)
        \right\rrbracket,
        \sum_{\substack{
            y_1, \ldots, y_\rho: \\
            (x_i, y_i) \in E \forall i
        }} \left\llbracket
            \asg^{(\ttt)} \text{\footnotesize{ satisfies }} (x_k, y_k)
        \right\rrbracket
     \right\}.
\end{aligned}
\end{align}
Consequently,
the value of $\asg^\bullet$ can be bounded as follows:

\begin{align}
\begin{aligned}
    & \val_{G^{\otimes \rho}}(\asg^\bullet) \\
    & \geq \frac{1}{|E(G^{\otimes \rho})|}
    \sum_{x_1, \ldots, x_\rho}
    \min\left\{
        \sum_{\substack{
            y_1, \ldots, y_\rho: \\
            (x_i, y_i) \in E \forall i
        }} \left\llbracket
            \asg^{(\ttt-1)} \text{\footnotesize{ satisfies }} (x_k, y_k)
        \right\rrbracket,
        \sum_{\substack{
            y_1, \ldots, y_\rho: \\
            (x_i, y_i) \in E \forall i
        }} \left\llbracket
            \asg^{(\ttt)} \text{\footnotesize{ satisfies }} (x_k, y_k)
        \right\rrbracket
     \right\} \\
    & = \frac{1}{|E|^{\rho}}
        \sum_{\substack{
            x_1,\ldots,x_{k-1},x_{k+1},\ldots,x_\rho,\\
            y_1,\ldots,y_{k-1},y_{k+1},\ldots,y_\rho:\\
            (x_i, y_i) \in E \forall i \neq k
        }}
    \sum_{x_k}
    \min\left\{
        \sum_{y_k: (x_k, y_k) \in E} \left\llbracket
            \asg^{(\ttt-1)} \text{\footnotesize{ satisfies }} (x_k, y_k)
        \right\rrbracket,
        \sum_{y_k: (x_k, y_k) \in E} \left\llbracket
            \asg^{(\ttt)} \text{\footnotesize{ satisfies }} (x_k, y_k)
        \right\rrbracket
    \right\} \\
    & = \min\left\{
        \frac{1}{|E|}  \sum_{(x,y) \in E} \left\llbracket
            \asg^{(\ttt-1)} \text{\footnotesize{ satisfies }} (x,y)
        \right\rrbracket,
        \frac{1}{|E|}  \sum_{(x,y) \in E} \left\llbracket
            \asg^{(\ttt)} \text{\footnotesize{ satisfies }} (x,y)
        \right\rrbracket
    \right\} \\
    & = \min\Bigl\{\val_G(\asg^{(\ttt-1)}), \val_G(\asg^{(\ttt)})\Bigr\},
\end{aligned}
\end{align}
which implies
$\val_{G^{\otimes \rho}}(\sqasg^{(k,\ttt)}) \geq \min\{\val_G(\asg^{(\ttt-1)}), \val_G(\asg^{(\ttt)})\}$.
Concatenating all reconfiguration sequences $\sqasg^{(k,\ttt)}$ for $\ttt \in [\TTT]$,
we obtain a reconfiguration sequence from
$\asg_\ini^{\otimes k} \otimes \asg_\tar^{\otimes \rho-k}$
to 
$\asg_\ini^{\otimes k-1} \otimes \asg_\tar^{\otimes \rho-k+1}$,
deriving that

\begin{align}
\begin{aligned}
    & \val_{G^{\otimes \rho}}(
        \asg_\ini^{\otimes k} \otimes \asg_\tar^{\otimes \rho-k} \reco
        \asg_\ini^{\otimes k-1} \otimes \asg_\tar^{\otimes \rho-k+1}
    ) \\
    & \geq
    \min_{1 \leq \ttt \leq \TTT}
    \val_{G^{\otimes \rho}}(
        \asg_\ini^{\otimes k-1} \otimes \asg^{(\ttt-1)} \otimes \asg_\tar^{\otimes \rho-k}
        \reco
        \asg_\ini^{\otimes k-1} \otimes \asg^{(\ttt)} \otimes \asg_\tar^{\otimes \rho-k}
    )
    \\
    & \geq
    \min_{1 \leq \ttt \leq \TTT}
    \val_{G^{\otimes \rho}}(\sqasg^{(k,\ttt)})
    \\
    & \geq
    \min_{0 \leq \ttt \leq \TTT} \val_G(\asg^{(\ttt)})
    \\
    & = \val_G(\sqasg) = \val_G(\asg_\ini \reco \asg_\tar).
\end{aligned}
\end{align}
Accordingly, we obtain
\begin{align}
    \val_{G^{\otimes \rho}}(
        \asg_\ini^{\otimes\rho} \reco \asg_\tar^{\otimes\rho}
    )
    \geq \min_{1 \leq k \leq \rho}
    \val_{G^{\otimes \rho}}(
        \asg_\ini^{\otimes k} \otimes \asg_\tar^{\otimes \rho-k} \reco
        \asg_\ini^{\otimes k-1} \otimes \asg_\tar^{\otimes \rho-k+1}
    )
    \geq \val_G(\asg_\ini \reco \asg_\tar),
\end{align}
completing the proof.
\end{proof}

\section{Polynomial-time Decidability for Projection Games}
\label{sec:proj}

We finally present a polynomial-time algorithm that decides
reconfigurability between a pair of satisfying labelings for a \emph{projection game}.

\begin{theorem}
\label{thm:proj-poly}
    For a satisfiable projection game $G = (X,Y,E,\Sigma,\Pi)$ and 
    its two satisfying labelings $\asg_\ini$ and $\asg_\tar$,
    we can decide in polynomial time if there is a 
    reconfiguration sequence from $\asg_\ini$ to $\asg_\tar$
    consisting of satisfying labelings for $G$.
\end{theorem}
\begin{proof}
Since each connected component of $G$ can be processed independently,
we only need to consider the case that $G$ is connected.
If $G$ consists of a single vertex (i.e., at least $X$ or $Y$ is empty),
the answer is always ``reconfigurable'' as it has no edges;
thus we can safely assume that $X \neq \emptyset$ and $Y \neq \emptyset$.
We show that $\asg_\ini$ and $\asg_\tar$ are reconfigurable to each other
if and only if
$\asg_\ini$ and $\asg_\tar$ agree on $X$.

We first show the ``only-if'' direction.
Suppose $\asg_\ini(x) \neq \asg_\tar(x)$ for some $x \in X$.
Starting from $\asg_\ini$,
changing the label of $x$ from $\asg_\ini(x)$ to any other label
must violate edges $e = (x,y)$ incident to $x$
since the label of $y$ should be mapped to $\asg_\ini(x)$ via $\pi_e$, as claimed.

We then show the ``if'' direction.
Suppose $\asg_\ini(x) = \asg_\tar(x)$ for every $x \in X$.
Then, it holds that
$\pi_e(\asg_\ini(y)) = \asg_\ini(x) = \asg_\tar(x) = \pi_e(\asg_\tar(y))$
for every edge $e=(x,y) \in E$.
We can thus obtain a reconfiguration sequence from $\asg_\ini$ to $\asg_\tar$ without breaking any constraint,
by changing the label of each vertex $y$ of $Y$
from $\asg_\ini(y)$ to $\asg_\tar(y)$ one by one, as desired. 
\end{proof}

\paragraph{Acknowledgments.}
I wish to thank the anonymous referees for their helpful suggestions.

\printbibliography

\end{document}